\documentclass[%
 reprint,
showpacs,preprintnumbers,
 amsmath,amssymb,
 aps,
pra,
]{revtex4-1}

\usepackage{graphicx}
\usepackage{dcolumn}
\usepackage{bm}
\usepackage{amssymb,amsfonts,latexsym,stmaryrd}

\usepackage{proof}
\usepackage{QED}
\usepackage{epsfig}
\usepackage{euscript}
\usepackage{tikz,ifdraft}
\usetikzlibrary{arrows,positioning,matrix,backgrounds,calc,fit}
\usepackage{url}

\newtheorem{theorem}{Theorem}[section]

\newtheorem{proposition}[theorem]{Proposition}

\newcommand{\beqa}{\begin{eqnarray*}}
\newcommand{\eeqa}{\end{eqnarray*}\par\noindent}

\newcommand{\XX}{\mathcal{X}}

\newcommand{\rarr}{\rightarrow}

\newcommand{\ie}{\textit{i.e.}~}

\newcommand{\Two}{\mathbf{2}}

\newcommand{\Real}{\mathbb{R}}

\newcommand{\IFF}{\; \Longleftrightarrow \;}

\newcommand{\card}[1]{|#1|}

\newcommand{\One}{\mathbf{1}}

\newcommand{\UU}{\mathcal{U}}

\newcommand{\XOR}{\mathsf{ONE}}

\newcommand{\MB}{\mathbf{M}}
\newcommand{\xx}{\mathbf{x}}
\newcommand{\vv}{\mathbf{v}}
\renewcommand{\emph}{\textbf}

\newcommand{\NN}{\mathcal{N}}
\newcommand{\yy}{\mathbf{y}}

\newcommand{\ww}{\mathbf{w}}

\newcommand{\Zero}{\mathbf{0}}
\newcommand{\prob}{\mathsf{Prob}}
\newcommand{\vphi}{\varphi}
\newcommand{\dt}{\delta^t}
\newcommand{\pij}{p_{i,j}}
\newcommand{\ev}{\eta^{\vv}}
\newcommand{\ew}{\eta^{\ww}}
\newcommand{\et}{\eta^{t}}
\newcommand{\pcf}{\psi}
\newcommand{\pU}{\pcf_{U}}
\newcommand{\tUs}{\theta^{U}_{s}}
\newcommand{\kUs}{k^{U}_{s}}
\newcommand{\pUs}{p^{U}_{s}}

\newcommand{\tU}{\theta_{U}}
\newcommand{\EP}{\mathcal{E}}
\newcommand{\kk}{\mathbf{k}}

\newcommand{\mijl}{m^i_{j,l}}
\newcommand{\Xij}{X^{i}_{j}}

\newcommand{\rr}{\mathbf{r}}

\begin{document}

\title{Logical Bell Inequalities}
\thanks{Financial support from the Perimeter Institute, EPSRC
Senior Research Fellowship EP/E052819/1, and the U.S. Office of Naval Research Grant Number
N000141010357 is gratefully acknowledged.}
\author{Samson Abramsky}
\email{samson@cs.ox.ac.uk}
\affiliation{Department of Computer Science, University of Oxford\\
Wolfson Building, Parks Road, Oxford OX1 3QD, U.K.}

\author{Lucien Hardy}
\email{lhardy@perimeterinstitute.ca}
\affiliation{Perimeter Institute, 31 Caroline Street North\\
Waterloo, ON N2L 2Y5, Canada}

\begin{abstract}
Bell inequalities play a central r\^ole in the study of quantum non-locality and entanglement, with many applications in quantum information.
Despite the huge literature on Bell inequalities, it is not easy to find a clear conceptual answer to what a Bell inequality \emph{is}, or a clear guiding principle as to how they may be derived.
In this paper, we introduce a notion of \emph{logical Bell inequality} which can be used to systematically derive testable inequalities for a very wide variety of situations.
There is a single clear conceptual principle, based on purely logical consistency conditions, which underlies our notion of logical Bell inequalities. We show that in a precise sense, \emph{all} Bell inequalities can be taken to be of this form.
Our approach is very general. It applies directly to \emph{any} family of sets of commuting observables. Thus it covers not only the $n$-partite scenarios to which Bell inequalities are standardly applied, but also Kochen-Specker configurations, and many other examples. 
There is much current work on experimental tests for \emph{contextuality}. Our approach directly yields, in a systematic fashion, testable inequalities for a very general notion of contextuality.

There has been much work on obtaining proofs of Bell's theorem `without inequalities' or `without probabilities'. These proofs are seen as  being in a sense more definitive and logically robust than the inequality-based proofs. On the hand, they lack the fault-tolerant aspect of inequalities.
Our approach reconciles these aspects, and in fact shows how the logical robustness can be converted into systematic, general derivations of inequalities with provable violations. Moreover, the kind of  strong non-locality or contextuality exhibited by the GHZ argument or by Kochen-Specker configurations can be shown to lead to \emph{maximal violations} of the corresponding logical Bell inequalities. Thus the qualitative and the quantitative aspects are combined harmoniously.
\end{abstract}

\pacs{03.65.Ud}

\maketitle

\section{Introduction}

There is a huge literature on Bell inequalities \cite{bell1964einstein,clauser1969proposed}, with many ingenious derivations of families of inequalities. However, a unifying principle with a clear conceptual basis has proved elusive.
In this paper, we introduce a form of  Bell inequality based on logical consistency conditions, which we call \emph{logical Bell inequalities}. This approach is both conceptually illuminating and technically powerful. 

To get some feeling for the results, we shall firstly discuss how Bell inequalities are used.
Their main application, of course, is to show the non-locality of quantum mechanics, as famously first demonstrated in Bell's theorem \cite{bell1964einstein}. More broadly, Bell inequalities are used to delineate those situations which can be accounted for by classical physical concepts from those which are inherently non-classical; and the content of Bell's theorem is exactly that quantum mechanics produces empirically accessible phenomena which fall into the latter category.

An important feature of the inequalities is that they have a \emph{fault-tolerant aspect} which makes them very suitable for experimental verification. Violation of a Bell inequality is quantitative, and allows non-classicality to be demonstrated without relying on idealized perfect measurements or state preparations.

There are also many applications of Bell inequalities in quantum information, for example in quantum key distribution \cite{ekert1991quantum,barrett2005no,acin2006bell}, quantum communication complexity \cite{brukner2004bell} and detection of quantum entanglement \cite{terhal2002detecting}, so that they also play a leading r\^ole in more applied work.

Although a huge literature on Bell inequalities has appeared over the past few decades, it is not easy to distill from this literature a clear conceptual answer to what a Bell inequality \emph{is}, or a clear guiding principle as to how they may be derived.

The present paper addresses this point, and introduces a notion of \emph{logical Bell inequality} which can be used to systematically derive testable inequalities for a very wide variety of situations.
The following points in particular are worth emphasizing:

\begin{itemize}
\item There is a single clear conceptual principle, based on purely logical consistency conditions, which underlies our notion of logical Bell inequalities. We show that in a precise sense, \emph{all} Bell inequalities can be taken to be of this form.

\item Our approach is very general --- much more so than the great majority of the literature on Bell inequalities. It applies directly to \emph{any} family of sets of commuting observables. Thus it covers not only the $n$-partite scenarios to which Bell inequalities are standardly applied, but also Kochen-Specker configurations, and many other examples. This is important since there is much current work on experimental tests for \emph{contextuality}, e.g.~\cite{bartosik2009experimental,kirchmair2009state}, a broader phenomenon than non-locality. Our approach directly yields, in a systematic fashion, testable inequalities for a very general notion of contextuality.

\item There has been much work on obtaining proofs of Bell's theorem `without inequalities' or `without probabilities' \cite{greenberger1990bell,hardy1992quantum,zimba1993bell}. These proofs are seen as  being in a sense more definitive and logically robust than the inequality-based proofs. On the hand, they lack the fault-tolerant aspect of inequalities.
Our approach fully reconciles these aspects, and in fact shows how the logical robustness can be converted into systematic, general derivations of inequalities with provable violations. Moreover, the kind of  strong non-locality or contextuality exhibited by the GHZ argument or by Kochen-Specker configurations can be shown to lead to \emph{maximal violations} of the corresponding logical Bell inequalities. Thus the qualitative and the quantitative aspects are combined harmoniously.
\end{itemize}

We now turn to a more precise, technical summary of our results.

We show that a rational inequality is satisfied by all non-contextual models if and only if it is equivalent to a logical Bell inequality. Thus quantitative tests for contextuality or non-locality always hinge on purely logical consistency conditions. We obtain explicit descriptions of complete sets of  inequalities for the convex polytope of non-contextual probability models, and the derived polytope of expectation values for these models. Moreover, these results are obtained at a high level of generality; they apply not only to the familiar cases of Bell-type scenarios, for any number of parties, but to all Kochen-Specker configurations, and in fact to \emph{any} family of sets of compatible measurements.
This generality is achieved by working with
\emph{measurement covers}, following the sheaf-theoretic approach to non-locality and contextuality introduced by the first author and Adam Brandenburger in \cite{abramsky2011unified}. 

We also obtain results for a number of special cases. We show that a model achieves maximal violation of a logical Bell inequality if and only if it is strongly (or maximally) contextual. We show that all Kochen-Specker configurations lead to maximal violations of logical Bell inequalities in a state-independent fashion. We also derive specific violations of logical Bell inequalities for models which are possibilistically contextual, meaning that they admit logical proofs of contextuality. Well-known examples of such models are those arising 
from a construction given by one of us twenty years ago \cite{hardy1992quantum, hardy1993nonlocality}.  

Inspiration for the present work was drawn from \cite{hardy1991new}, which derives  some particular cases of logical Bell inequalities.
Developing these ideas in the general setting provided by  \cite{abramsky2011unified} proves to be fruitful, and indicates the potential for a structural approach to quantum foundations.

 \subsection{A Simple Observation}
We begin with a simple and very general scenario.

Suppose we have propositional formulas $\vphi_1, \ldots , \vphi_N$.
We suppose further that we can assign a probability $p_i$ to each $\vphi_i$.

In particular, we have in the mind the situation where the boolean variables appearing in $\vphi_i$ correspond to empirically testable quantities; $\vphi_i$ then expresses a condition on the outcomes of an
experiment involving these quantities.
The probabilities $p_i$ are obtained from the statistics of these experiments.

Now let $P$ be the probability of $\Phi := \bigwedge_i \vphi_i$.
Using elementary probability theory, we can calculate:
 \[ \begin{array}{lclclcl}
 1 - P & = & \prob(\neg \Phi) & = & \prob(\bigvee_i \neg \vphi_i) & \leq & \sum_i \prob(\neg \vphi_i) \\
 & & & = & \sum_i (1 - p_i) & = & N - \sum_i p_i . 
 \end{array}
 \]
 Tidying this up yields $\sum_i p_i \; \leq \; N-1+P$.

Now suppose that the formulas $\vphi_i$ are \emph{jointly contradictory}; \ie $\Phi$ is unsatisfiable.
This implies that $P = 0$. Hence we obtain the inequality
 \[ \sum_i p_i \; \leq \; N-1. \]
This inequality was obtained in \cite{hardy1991new}, where it was used to derive chained Bell inequalities (as originally obtained in \cite{braunstein1990wringing}).   It is an example of a \emph{logical Bell inequality}.  In Section~\ref{generalformsection} we shall give a general form for logical  Bell inequalities.  

 \subsection{A Curious Observation}

Quantum Mechanics tells us that we can find propositions $\vphi_i$ describing outcomes of certain measurements, which not only \emph{can} but \emph{have} been performed.
From the observed statistics of these experiments, we have \emph{very} highly confirmed probabilities $p_i$.
These propositions are easily seen to be jointly contradictory.
Nevertheless, the inequality
\[ \sum_i p_i \; \leq \; N - 1 \]
is observed to be \emph{strongly violated}. In fact, the maximum violation of $1$ can be 
achieved \footnote{Since each $p_i$ is a probability, its maximum value is $1$, so the `algebraic maximum' of the sum $\sum_i p_i$ is $N$.}.

How can this be?

The best resolution to this puzzle on offer is that each formula $\vphi_i$ involves a proper subset $X_i$ of the total set $X$ of boolean variables which appear in the family, and hence in the conjunction $\Phi$. 
There is no global assignment of probabilities to all the variables $X$ simultaneously which yields the empirically observed probabilities. Hence the ascription of a probability to $\Phi$ is the invalid step.
This is given general mathematical meaning in terms of an obstruction to the existence of a global section in \cite{abramsky2011unified,abramsky2011cohomology}, extending 
\cite{fine1982hidden}.

This does seem an uncomfortably slender basis on which to defend logical consistency, since it seems hard to avoid the conclusion that the null event should be assigned probability $0$.

This argument can be seen as a theory-independent derivation of the impossibility of measuring all the variables in $X$ simultaneously, even in principle, on pain of a direct clash between logical consistency and empirical evidence.
We simply cannot regard the variables as each representing a global, context-independent quantity.

\subsection{Logical Bell and CHSH inequalities}
We shall call the inequality
\[ \sum_i p_i \; \leq \; N - 1 \]
a \emph{logical Bell inequality}.
We can also derive an associated inequality for expectations.

We shall associate truth of a formula with the value $+1$, and falsity with $-1$.
We then have the expected value $E_i$ of the formula $\vphi_i$ given by
\[ E_i = (+1)\cdot p_i + (-1)\cdot (1-p_i) = 2 p_i - 1. \]
From the Bell inequality, we obtain:
\[ \sum_i E_i = \sum_i (2p_i -1) = 2 \sum_i p_i - N \leq 2(N-1) - N = N - 2. \]
Moreover, if $K$ is an upper bound as the expectations range over probability assignments, $-K$ must be a lower bound, as we can substitute $1 - p_i$ for $p_i$ to get the expected value $-E_i$.
Thus this is a bound on the absolute value of the expectations, so we obtain the \emph{logical CHSH inequality}:
\begin{equation}
\label{chsheq}
 |\sum_i E_i | \; \leq \; N - 2 .
\end{equation}

Note that these inequalities are very general, and independent of any particular setting.
We shall now show how they apply to familiar scenarios arising from quantum mechanics and the study of non-locality.

\section{Probabilistic models of experiments}

Our general setting will be the probability models commonly studied in quantum information and quantum foundations \footnote{We shall eventually consider a much more general form of such models introduced in \cite{abramsky2011unified}, which allow a uniform treatment of contextuality, including Kochen-Specker configurations etc.}.
In these models, a number of agents each has the choice of one of several measurement settings; and each measurement has a number of distinct outcomes. For most of this paper, we shall focus on measurements with two possible outcomes; however, we will show how our results can be extended to measurements with multiple outcomes in Section~\ref{moutsec}.
For each choice of a measurement setting by each of the agents, we have a probability distribution on the joint outcomes of the measurements.

For example, consider the following tabulation of such a model.

\begin{center}
\begin{tabular}{l|ccccc}
& $(0, 0)$ & $(1, 0)$ & $(0, 1)$ & $(1, 1)$  &  \\ \hline
$(a, b)$ & $1/2$ & $0$ & $0$ & $1/2$ & \\
$(a, b')$ & $3/8$ & $1/8$ & $1/8$ & $3/8$ & \\
$(a', b)$ & $3/8$ & $1/8$ & $1/8$ & $3/8$ &  \\
$(a', b')$ & $1/8$ & $3/8$ & $3/8$ & $1/8$ &
\end{tabular}
\end{center}

Here we have two agents, Alice and Bob. Alice can choose from the settings $a$ or $a'$, and Bob can chooose from $b$ or $b'$. These choices correspond to the rows of the table. The columns correspond to the joint outcomes for a given choice of settings by Alice and Bob, the two possible outcomes for each individual measurement being represented by $0$ and $1$.
The numbers along each row specify a probability distribution on these joint outcomes.

\subsection{The Bell Model}
\label{bellmodsec}

A standard version of Bell's theorem uses the probability table given above. This table can be realized in quantum mechanics, e.g.~by a Bell state, written in the $Z$ basis as
\[ \frac{\mid \uparrow \uparrow  \rangle \; + \; \mid \downarrow \downarrow \rangle}{\sqrt{2}} , \]
subjected to spin measurements in the   $XY$-plane of the Bloch sphere, at a relative angle of $\pi/3$.

\subsubsection*{Logical analysis of the Bell table}
We now pick out a subset of the elements of each row of the table, as indicated in the following table.
\begin{center}
\begin{tabular}{l|ccccc}
& $(0, 0)$ & $(1, 0)$ & $(0, 1)$ & $(1, 1)$  &  \\ \hline
$(a, b)$ & { \fcolorbox{gray}{yellow}{1/2}} & $0$ & $0$ & { \fcolorbox{gray}{yellow}{1/2}} & \\
$(a, b')$ &  { \fcolorbox{gray}{yellow}{3/8}} & $1/8$ & $1/8$ & { \fcolorbox{gray}{yellow}{3/8}} & \\
$(a', b)$ & { \fcolorbox{gray}{yellow}{3/8}} & $1/8$ & $1/8$ & { \fcolorbox{gray}{yellow}{3/8}} &  \\
$(a', b')$ & $1/8$ & { \fcolorbox{gray}{yellow}{3/8}} & { \fcolorbox{gray}{yellow}{3/8}} & $1/8$ &
\end{tabular}
\end{center}

If we read $0$ as true and $1$ as false, the highlighted positions in the table are represented by the following propositions:
\[ \begin{array}{rcccccccc}
\vphi_1 & = &  a \wedge b & \vee & \neg a \wedge \neg b & = & a & \leftrightarrow & b \\
\vphi_2 & = & a \wedge b' & \vee & \neg a \wedge \neg b' & = &  a & \leftrightarrow & b' \\
\vphi_3 & = & a' \wedge b & \vee & \neg a' \wedge \neg b & = &  a' & \leftrightarrow & b \\
\vphi_4 & = & \neg a' \wedge b' & \vee & a' \wedge \neg b' & = & a' & \oplus & b' .
\end{array}
\]
The first three rows are the correlated outcomes; the fourth is anticorrelated.
These propositions are easily seen to be contradictory. Indeed, starting with
$\vphi_4$, we can replace $a'$ with $b$ using $\vphi_3$, $b$ with $a$ using $\vphi_1$, and $a$ with $b'$ using $\vphi_2$, to obtain $b' \oplus b'$, which is obviously unsatisfiable.

We see from the table that $p_1 = 1$, $p_i = 6/8$ for $i = 2, 3, 4$.
Hence the violation of the Bell inequality is $1/4$; and of the CHSH inequality $1/2$.

We may note that the logical pattern shown by this jointly contradictory family of propositions underlies the familiar CHSH correlation function.

\subsection*{Some notation}
Later we will develop some notation for the general case. To prepare the way we will indicate how this notation will work in the Bell model.  
First, in this case, we put
\[ X=\{ a, b, a', b'\} \] 
as this is the set of boolean variables we are interested in.  Next, we consider subsets $U \subseteq X$, corresponding to the different combinations of measurements  we might perform --- the \emph{measurement contexts}. One such subset is $U=\{a,b\}$.  We denote the set of all such subsets by $\mathcal{U}$.  Thus, in this case, we have
\[ \mathcal{U}=\{ \{a,b\}, \{a,b'\}, \{a',b\}, \{a',b'\}  \}  . \]
A basic measurement such as $a$ has possible outcomes $0$ or $1$. We shall write $\Two := \{ 0, 1 \}$ for the set of possible outcomes. A joint outcome for a set of measurements $U$ can be specified by a function $s : U \rarr \Two$. For example, if we perform the measurements in $U = \{ a, b \}$, and $a$ has outcome $0$ and $b$ has outcome $1$, this is described by the function
\[ \{ a \mapsto 0, \; b \mapsto 1 \}  \]
which maps $a$ to $0$ and $b$ to $1$.
This function corresponds to the cell in the first row and third column of the Bell table.
The set of all such functions is denoted by $\Two^U$. Thus for $U=\{a,b\}$,
\[ \Two^U = \{  f_{ij} \mid i, j = 0, 1 \} \]
where $f_{ij} = \{ a \mapsto i, b \mapsto j \}$.
This corresponds to the set of cells in the first row of the table.

A probability model such as the Bell table shown above is given by specifying a probability distribution $d_U$ on $\Two^U$ for each $U \in \UU$. Thus $d_U$ is a function $d_U : \Two^U \rarr [0, 1]$ such that $\sum_{s \in \Two^U} d_U(s) = 1$.
These distributions correspond to the rows of the Bell table.

The proposition $\vphi_1$ pertains to the context $U=\{a, b\}$; note that it only uses the variables in $U$.   We can think of $\Two^U$ as the set of truth-value assignments to the boolean variables in $U$, where we interpret $0$ as true and $1$ as false. The set of satisfying assignments for the formula $\vphi_1$ --- the subset of $\Two^U$ for which this proposition is true ---  is
\[ S(U) = \{   \{ a \mapsto 0, \; b \mapsto 0 \}, \;\; \{ a \mapsto 1, \; b \mapsto 1 \} \}  .  \]
We have given such a proposition $\vphi_i$ for each element of $\mathcal{U}$.  The highlighted items in the $i$'th row of the table  form the set $S(U_i)$ of satisfying assignments for $\vphi_i$, where $U_i$ is the corresponding measurement context.

\subsection{A bipartite logical model}
\label{bipartlogmodel}
We know turn to the model introduced by one of us in 1992 \cite{hardy1992quantum, hardy1993nonlocality}.
The original purpose of this construction was to show a `logical' proof of Bell's theorem in the bipartite case, following the GHZ tripartite construction.
Reflecting this, we shall  only need to consider the \emph{support table} of the model to demonstrate a violation of the inequalities.

Consider, for example, the following table, which has a quantum realization as described in \cite{hardy1993nonlocality}.
\begin{center}
\begin{tabular}{l|cccc}
 &  $(0, 0)$ & $(1, 0)$ & $(0, 1)$ & $(1, 1)$ \\ \hline
$(a, b)$ &  $1$ &  $1$ &  $1$ &  $1$ \\
$(a', b)$ &   $0$ &  $1$ &  $1$ &  $1$ \\
$(a, b')$ &  $0$ &  $1$ &  $1$ & $1$ \\
$(a', b')$ &   $1$ &  $1$ &  $1$ & $0$ \\
\end{tabular}
\end{center}
This table has a $1$ for every entry in the  model with a positive probability.

If we interpret outcome $0$ as true and $1$ as false, then the following formulas all have positive probability:
\[ a \wedge b, \quad \neg (a \wedge b'), \quad \neg (a' \wedge b), \quad  a' \vee b'. \]
However, these formulas are not simultaneously satisfiable.

Note that the formulas $\vphi_i$ for $i = 2,3,4$ describe the full support of this model for the corresponding rows; hence $p_2 = p_3 = p_4 = 1$.
It follows that the model achieves a violation of $p_1 = \prob(a \wedge b)$ for the Bell inequality, and a violation of  $2p_1$ for the CHSH inequality.

Note that this calculation can be made purely on the basis of the support table.

\section{The general case: structure of supports}

We now turn to a general analysis. The setting will be that of \cite{abramsky2011unified}, but we shall develop what we need in a self-contained fashion.

We shall begin by looking just at the \emph{supports} of probability models, which suffice to describe many forms of contextual and non-local behaviour, as we have already illustrated with the model described in Section~\ref{bipartlogmodel}. We shall then go on to look at generalized probability models themselves.

\subsection*{Notation} We shall use the notation introduced in the previous section: we define $\Two := \{ 0, 1 \}$, and write $\Two^U$ for the set of all functions from a set $U$ into $\Two$.  
We shall also the following notation for function restriction. If $s : X \rarr \Two$ is a function, and $U \subseteq X$, then we write $s | U : U \rarr \Two$ for the restriction of $s$ to $U$.
For example, if $X=\{ a, b, a', b'\}$, $U = \{ a, b \}$, and $s : X \rarr \Two$ is the function
\[ \{ a \mapsto 0, \; b \mapsto 1, \; a' \mapsto 1, \; b' \mapsto 0 \} \]
then $s | U$ is the function
\[ \{ a \mapsto 0, \; b \mapsto 1 \} . \]

\subsection{Structure of support tables}
We fix a set of boolean variables $X$, and a \emph{cover} $\UU$, \ie a family of subsets of $X$ such that $\bigcup \UU = X$.

A \emph{probability model} on a cover $(X, \UU)$ is a family $\{d_U \}_{U \in \UU}$, where $d_U$ is a probability distribution on $\Two^U$. 

We think of the sets $U \in \UU$ as the \emph{compatible sets of measurements}, which index the `rows' of the probability table. Given such a row $U$, $\Two^U$ is the set of possible joint outcomes of these measurements. The distribution $d_U$ gives the probability for each such joint outcome.

The \emph{support} of the model at $U \in \UU$ is the set $S(U) \subseteq \Two^U$ of those $s \in \Two^U$ such that $d_U(s) > 0$.

A \emph{global section} for the support of the model is  an assignment
\[ s : X \rarr \Two \]
such that, for all $U \in \UU$, $s | U \in S(U)$ \footnote{The support of a probability model can be given the structure of  a \emph{presheaf}, in such a way that the above definition corresponds to the usual notion of global section. See  \cite{abramsky2011unified} for details.}.

We can think of global sections in geometric terms, as coherently gluing together a family of local sections $s_U \in S(U)$, indexed by $U \in \UU$.
This geometrical idea of global section can be related to logical notions.
A formula $\vphi_U$ over a set of variables $U \in \UU$ has a set of satisfying assignments which is a subset of $\Two^U$. Note that, if $U$ is finite, \emph{any} subset of $\Two^U$ can be defined in this way by a propositional formula.
For each $U \in \UU$, let $\vphi_U$ be a formula whose set of satisfying assignments is $S(U)$.
Global sections correspond precisely to satisfying assignments for the formula
\[ \vphi \; = \; \bigwedge_{U \in \UU} \vphi_{U} . \]

As shown in detail in \cite{abramsky2011unified}, the existence of global sections provides a canonical form for \emph{non-contextual hidden-variable theories}.

We can define a probabilistic model to be \emph{possibilistically noncontextual} \cite{abramsky2011unified} if for every element $s \in S(U)$ of its support, there is a global section $s'$ such that $s' | U = s$.
If this does not hold, the model is \emph{contextual}, or in particular \emph{non-local}.
In fact, as shown in \cite{abramsky2011unified}, this form of contextuality or non-locality is strictly stronger than the usual probabilistic notions. For example, the Bell model studied in the previous section is non-local, but is in fact possibilistically non-contextual.
The possibilistically contextual models are those which admit logical proofs of Bell's theorem: `Bell's theorem without inequalities' \cite{greenberger1990bell}.

We can now give a completely general argument that for any model which is contextual in this strong possibilistic sense, we can obtain a violation of instances of the generalized Bell and CHSH inequalities.

\begin{proposition}
Any possibilistically contextual model violates a logical Bell/CHSH inequality.
\end{proposition}
\begin{proof}
Suppose that a model is possibilistically contextual, with $s \in S(U)$ such that there is no global section for $S$ restricting to $s$.
We define a formula $\vphi_s$, describing $s$. This formula can be written explicitly as
\begin{equation}
\label{phiseq}
\vphi_s \;\; :=  \;\; \bigwedge_{s(x) = 0} x \;\; \wedge \;\; \bigwedge_{s(x) = 1} \neg x . 
\end{equation}
The only satisfying assignment for $\vphi_s$ in $\Two^U$  is $s$.

For all $U' \in \UU$ with $U' \neq U$, we define $\vphi_{U'}$ to be a formula which defines the support of the model on the `row' $U'$. Explicitly, we can define:
\[ \vphi_{U'} \;\; := \;\; \bigvee_{s' \in S(U')} \vphi_{s'} . \]

The fact that there is no global section on the support which restricts to $s$
says exactly that the formula $\vphi_s \wedge \bigwedge_{U' \neq U} \vphi_{U'}$ is not satisfiable.
Since $p_{U'} = 1$ for $U \neq U' \in \UU$, the Bell inequality with respect to these formulas is violated by $p_{\vphi_{s}} = p(s) > 0$,
while violation of the CHSH inequality is by $2p(s)$.
\end{proof}

\subsection{Strong Contextuality}

A still stronger form of contextuality is identified in \cite{abramsky2011unified}. A model is defined to be \emph{strongly contextual} if its support has no global section; equivalently, the propositional formulas defining its support are not simultaneously satisfiable.

It is shown in \cite{abramsky2011unified} that all $n$-partite states GHZ($n$), for $n \geq 3$, are strongly contextual in this sense. It is also shown that strong contextuality is equivalent to the model being maximally contextual, in the sense of having no non-trivial convex decomposition into a non-contextual model and a no-signalling model.

We now have the following result.

\begin{proposition}
\label{scprop}
A model achieves maximal violation of a logical Bell inequality if and only if it is strongly contextual.
\end{proposition}
\begin{proof}
Suppose that the model is strongly contextual. For each row $U$, we can define the formula $\vphi_U$ corresponding to the support of the model on that row, as in the proof of the previous proposition.
Since the probability of each $\vphi_U$ is $1$, we obtain the maximum violation of $1$.

For the converse, if maximal violation is achieved, there are a family of rows $U_1, \ldots , U_N$, and propositions $\vphi_{i}$ defining subsets $S(U_{i}) \subseteq \Two^{U_i}$, such that $\bigwedge_i \vphi_{i}$ is unsatisfiable, and $\sum_i p_i = N$. This implies that $p_i = 1$ for all $i$, and hence that $S(U_i)$ contains the support of the model on $U_i$. The unsatisfiability of $\bigwedge_i \vphi_{i}$ means that there is no global section which restricts to each $S(U_i)$, which means \textit{a fortiori} that the model is strongly contextual.
\end{proof}

\subsubsection*{Example:~the GHZ state}
We consider the tripartite GHZ state \cite{greenberger1989going,greenberger1990bell}, which we write in the $Z$ basis as
\[ \frac{\mid \uparrow \uparrow \uparrow  \rangle \; + \; \mid \downarrow \downarrow \downarrow \rangle}{\sqrt{2}} , \]
with $X$ and $Y$ measurements  in each component.
The relevant part of the support table for the resulting probability model can be specified as follows:
\begin{center}
\begin{tabular}{c|cccccccc}
&  $000$ & $001$ & $010$ & $011$  & $100$ & $101$ & $110$ & $111$  \\ \hline
$abc$ & $1$ & $0$ & $0$ & $1$ & $0$ & $1$ & $1$ & $0$  \\
$ab'c'$ & $0$ & $1$ & $1$ & $0$ & $1$ & $0$ & $0$ & $1$  \\
$a'bc'$ & $0$ & $1$ & $1$ & $0$ & $1$ & $0$ & $0$ & $1$  \\
$a'b'c$ & $0$ & $1$ & $1$ & $0$ & $1$ & $0$ & $0$ & $1$ 
\end{tabular}
\end{center}
Given boolean variables $x$, $y$, $z$, we define
\begin{equation}
\label{pdef}
\Psi_{xyz} \; := \; \neg x \oplus \neg y \oplus \neg z .
\end{equation}
The support for each row can be specified by the following formulas:
\[ \vphi_1 := \neg \Psi_{abc}, \;\; \vphi_{2} := \Psi_{ab'c'},  \;\; \vphi_3 := \Psi_{a'bc'},  \;\; \vphi_4 := \Psi_{a'b'c} . \]
It can be verified that these formulas are not simultaneously satisfiable; in fact, such a verification is what the well-known argument by Mermin in terms of `instruction sets' \cite{mermin1990quantum} amounts to.

Thus the tripartite GHZ  state maximally violates a logical Bell inequality. Similar arguments apply to $n$-partite GHZ states for all $n > 3$; see \cite{abramsky2011unified}.

\subsubsection*{Example:~the PR box}
We consider the Popescu-Rohrlich box \cite{popescu1994quantum}, which achieves super-quantum correlations while respecting no-signalling.

\begin{center}
\begin{tabular}{l|ccccc}
& $(0, 0)$ & $(0, 1)$ & $(1, 0)$ & $(1, 1)$  &  \\ \hline
$(a, b)$ & $1$ & $0$ & $0$ & $1$ & \\
$(a, b')$ & $1$ & $0$ & $0$ & $1$ & \\
$(a', b)$ & $1$ & $0$ & $0$ & $1$ & \\
$(a', b')$ & $0$ & $1$ & $1$ & $0$ & 
\end{tabular}
\end{center}

The supports of the rows of this table are specified by the following formulas:
\[ a \leftrightarrow b, \qquad a  \leftrightarrow b', \qquad a'  \leftrightarrow b, \qquad a' \oplus b' \]
which are not simultaneously satisfiable. Thus this model maximally violates a logical Bell inequality.

Note that these formulas are the same as those we used for the Bell model in section~\ref{bellmodsec}.
In this case, however, they cover the whole support of the model, corresponding to the fact that the PR-box attains the algebraic maximum of the CHSH correlation function.

\subsection{Kochen-Specker configurations}

The notion of model we are considering, following \cite{abramsky2011unified},  is much more general than the usual `Bell scenarios'.
For example, any set $\XX$ of quantum observables  gives rise to a cover in our sense, where the sets in the cover correspond to the maximal compatible subsets of $\XX$.
Since we are currently restricting our attention to two-outcome measurements, we shall only consider dichotomic observables.
If we fix a state, then for each maximal set of compatible observables, \ie each row of the table, we get a probability distribution on joint outcomes of the observables in the family, following the usual quantum mechanical recipe. The details are spelled out in \cite{abramsky2011unified}.

The usual Bell case arises when the observables are partitioned according to the sites or parties; the sets in the cover correspond to a choice of one observable from each site, represented on a tensor product in the standard fashion.

 Equally, however, any Kochen-Specker configuration gives rise to a cover in our sense \footnote{For further details on this perspective on the Kochen-Specker theorem, see \cite{abramsky2011unified}.}.
 Given a family of unit vectors representing distinct rays in $\Real^d$, we consider the corresponding dichotomic observables, whose spectral resolutions project onto the ray and its orthogonal. We shall label the outcome corresponding to the ray as $0$, and the orthogonal outcome as $1$.

These observables are compatible if and only if the corresponding rays are orthogonal. Thus the maximal compatible families correspond to the families of vectors which determine orthonormal bases of $\Real^d$.
It follows that, for any quantum state, the only possible outcomes for one of these maximal compatible families are those where exactly one of the outcomes is labelled $0$. Thus for any state, the support of the probability model it gives rise to satisfies the following formula for each set $U$ in the cover:
\[ \XOR(U) \; := \; \bigvee_{x \in U} (x \; \wedge \; \bigwedge_{x' \in U \setminus \{ x \}} \neg x') . \]
The essential property of Kochen-Specker configurations is exactly that there is no global section for this family of supports; or equivalently, that the formula
\[ \bigwedge_{U \in \UU} \XOR(U)  \]
is unsatisfiable. It follows immediately that, given a Kochen-Specker configuration, the probability model generated by \emph{any} quantum state with respect to the corresponding family of observables is strongly contextual.
This fully explicates the state-independent nature of the Kochen-Specker theorem.

Hence we obtain the following corollary to Proposition~\ref{scprop}.
\begin{proposition}
For any Kochen-Specker configuration, and for any quantum state, the corresponding probability model maximally violates a logical Bell inequality.
\end{proposition}

Thus we have a perfectly general way of obtaining experimentally  testable inequalities, with maximal violations, from any Kochen-Specker configuration.

\subsubsection*{Example:~the 18-vector configuration in $\Real^4$}
We look at the 18-vector construction in $\Real^4$ from \cite{cabello1996bell}.
This uses the following measurement cover $\UU = \{ U_1 , \ldots , U_9 \}$, where the columns $U_i$ are the sets in the cover.
\begin{center}
\begin{tabular}{|c|c|c|c|c|c|c|c|c|} \hline
$U_1$ & $U_2$ &  $U_3$ & $U_4$ & $U_5$ & $U_6$ & $U_7$ & $U_8$ & $U_9$ \\ \hline\hline
$A$ & $A$ & $H$ & $H$ & $B$ & $I$ & $P$ & $P$  & $Q$ \\ \hline
$B$ & $E$ & $I$ & $K$ & $E$ & $K$ & $Q$ & $R$ & $R$  \\ \hline
$C$ & $F$ & $C$ & $G$ & $M$ &  $N$ & $D$ & $F$ & $M$  \\ \hline
$D$ & $G$ & $J$ & $L$ & $N$  & $O$ & $J$ & $L$ & $O$  \\ \hline
\end{tabular}
\end{center}

The standard argument that this is a Kochen-Specker configuration \cite{cabello1996bell,abramsky2011unified} amounts to verifying that the formula
\[ \bigwedge_{i=1}^9 \XOR(U_i)  \]
is unsatisfiable.
Thus for any quantum state, the resulting probability model will maximally violate a logical Bell inequality.

\subsubsection*{Example:~the Peres-Mermin Square}
We look at an important example, the Peres-Mermin square \cite{peres1990incompatible,mermin1990simple}, which can be realized in quantum mechanics using two-qubit observables.

The structure of the square is as follows:
\begin{center}
\begin{tabular}{|c|c|c|}
\hline
$A$ & $B$ & $C$ \\ \hline
$D$ & $E$ & $F$ \\ \hline
$G$ & $H$ & $I$ \\ \hline
\end{tabular}
\end{center}
The compatible families of measurements are the rows and columns of this table.
The  key property differs from the usual Kochen-Specker situation
in that we don't ask for exactly one $1$ at each maximal context. Instead, we ask
that each `row context' has an odd number of $1$'s whereas each `column context'
has an even number of $1$'s.
Hence the support table is the following.
\begin{center}
\begin{tabular}{l|cccccccc}
      & $000$ & $001$ & $010$ & $011$ & $100$ & $101$ & $110$ & $111$   \\ \hline
$ABC$ &  $0$  &  $1$  &  $1$  &  $0$  &  $1$  &  $0$  &  $0$  &  $1$  \\
$DEF$ &  $0$  &  $1$  &  $1$  &  $0$  &  $1$  &  $0$  &  $0$  &  $1$  \\
$GHI$ &  $0$  &  $1$  &  $1$  &  $0$  &  $1$  &  $0$  &  $0$  &  $1$  \\
$ADG$ &  $1$  &  $0$  &  $0$  &  $1$  &  $0$  &  $1$  &  $1$  &  $0$  \\
$BEH$ &  $1$  &  $0$  &  $0$  &  $1$  &  $0$  &  $1$  &  $1$  &  $0$  \\
$CFI$ &  $1$  &  $0$  &  $0$  &  $1$  &  $0$  &  $1$  &  $1$  &  $0$  \\
\end{tabular}
\end{center}
Note that the first three lines correspond to the row contexts and the remaining
three to the column contexts from the square.

The following formulas characterize the supports for each line of the table:
\[ \begin{array}{ccccccccc}
\vphi_{1} & := & \Psi_{ABC},  & \vphi_2 & := & \Psi_{DEF},  & \vphi_3 & := & \Psi_{GHI} \\
\vphi_4 & := & \neg \Psi_{ADG}, & \vphi_5 & := & \neg \Psi_{BEH}, & \vphi_6 & := & \neg \Psi_{CFI} . 
\end{array}
\]
Here we use $\Psi_{xyz}$ as defined in~(\ref{pdef}).

It can be verified that these formulas are not simultaneously satisfiable. 
Thus the Peres-Mermin square maximally violates a logical Bell inequality.

\section{General probabilistic models}
Suppose we are given a cover $\UU$ on a set $X$.
A general probability model over $\UU$ assigns a probability distribution $d_U$ on the set $\Two^U$ for each $U \in \UU$ \footnote{In this note, following \cite{abramsky2011unified}, we only consider finite sets of measurements, and thus discrete probability distributions suffice.}.

Each global assignment $t \in \Two^X$ induces a deterministic probability model $\dt$:
\[ \dt_U(s) = \left\{ \begin{array}{ll}
1, & t|U = s \\
0 & \mbox{otherwise.}
\end{array}
\right.
\]

We have the following result from \cite[Theorem 8.1]{abramsky2011unified}:
\begin{theorem}
\label{ncdetth}
 A probability model $\{ d_U \}$ is non-contextual if and only if it can be written as a convex combination $\sum_{j \in J} \mu_j \delta^{t_j}$ where $t_j \in \Two^X$ for each $j \in J$.
This means that for each $U \in \UU$, 
\[ d_U = \sum_j \mu_j \delta^{t_j}_{U} . \]
\end{theorem}
In fact, this gives a canonical form for such models, subsuming the usual notions of local or non-contextual hidden-variable models.

\section{The general form of logical Bell inequalities}
 \label{generalformsection}
It will be useful to establish some notation for expressing logical Bell inequalities.
Suppose we are given a cover $\UU$ on a set $X$. As illustrated in the examples we have looked at previously, we will regard $X$ as a set of boolean variables. We shall consider expressions of the form
\[  \sum_{i=1}^N k_i \vphi_i \]
where for each $i$, $k_i$ is a non-negative integer, and $\vphi_i$ is a formulas whose variables
are drawn from $U_i \in \UU$.

We think of such expressions as  \emph{multisets} of formulas, where $\vphi_i$ appears with multiplicity $k_i$.
A \emph{sub-multiset} of  $\sum_{i \in I} k_i \vphi_i$ is an expression of the form  $\sum_{i \in I} k'_i \vphi_i$, where for each $i$, $0 \leq k'_i \leq k_i$. The \emph{cardinality} of $\sum_{i \in I} k_i \vphi_i$ is $\sum_{i \in I} k_i$.
We say that $\sum_{i \in I} k_i \vphi_i$ is \emph{$K$-consistent} if for every sub-multiset of cardinality $> K$,  the underlying set of formulas with positive support has no satisfying assignment.

Given a positive integer $K$, we consider the expression
\begin{equation}
\label{bellineqexp}
\sum_{i=1}^N k_i p(\vphi_i) \; \leq \; K .
\end{equation}
If we are given a probability model $\{ d_U \}_{U \in \UU}$, we can evaluate the formal expression $p(\vphi_i)$ as $p_i := d_{U_i}(S_i)$, where $S_i$ is the set of satisfying assignments in $\Two^{U_i}$ for $\vphi_i$ --- \ie the event defined by $\vphi_i$.

The model \emph{satisfies} the expression~(\ref{bellineqexp}) if
\[ \sum_{i=1}^N k_i p_i \; \leq \; K . \]

\begin{proposition}
\label{kconprop}
The inequality~(\ref{bellineqexp}) is satisfied by all non-contextual models if and only if the multiset $\sum_{i \in I} k_i \vphi_i$ is $K$-consistent.
\end{proposition}
\begin{proof}
By Theorem~\ref{ncdetth}, any non-contextual model can be written as a convex combination $\sum_j \mu_j \delta^{t_j}$, where $t_j \in \Two^X$.

It suffices to verify~(\ref{bellineqexp}) for the deterministic models $\dt$, since if for each $j$ we have $\sum_{i} k_i \pij \leq K$, where $\pij = \delta^{t_j}_{U_i}(S_i)$, then:
\[ \sum_i k_i(\sum_j \mu_j \pij) \, = \, \sum_j \mu_j (\sum_i  k_i \pij) \, \leq \, \sum_j \mu_j K \, = \, K . \]
Now if the multiset $\sum_{i \in I} k_i \vphi_i$ is $K$-consistent, any $t \in \Two^X$, viewed as a boolean assignment on $X$, can satisfy a sub-multiset of cardinality at most $K$,
and hence $\dt$ will satisfy the inequality~(\ref{bellineqexp}).

Conversely, if $t$ satisfies $K+1$ formulas in the multiset, each corresponding term in~(\ref{bellineqexp}) will be assigned probability $1$ in $\dt$, and hence we will have $\sum_i k_i p_i \geq K+1$.
\end{proof}

Note that the form of logical inequality which we have considered previously is a special case, where we have a set of $N$ formulas which is $(N{-}1)$-consistent. Allowing for the more general notion of $K$-consistency
leads to sharper inequalities, which are needed to obtain completeness.

It is important to note that there is no requirement for the sets $U_i$  to be distinct. Thus different formulas occurring in the multiset may define overlapping subsets of the same row.

We define the general notion of logical Bell inequality over a cover $\UU$ to be given by expressions of the form~(\ref{bellineqexp}), where the multiset of formulas is $K$-consistent. Note that this class of inequalities is defined solely in terms of the cover $\UU$, and a purely logical condition on the formulas.
Thus we may indeed regard this as a logical class; the interesting point is that we can obtain quantitative information about contextuality from conditions which are derived in a purely logical fashion.

\section{Completeness of logical Bell inequalities}

We shall now show that logical Bell inequalities completely characterize contextuality.

We begin by recalling the definition of the \emph{incidence matrix} from  \cite{abramsky2011unified}.
Given a cover $\UU$ on a set $X$, we define a matrix $\MB$ whose rows are indexed by pairs $(U, s)$, where $U \in \UU$, and $s \in \Two^U$; and whose columns are indexed by global assignments $t \in \Two^X$. The matrix entries are defined by:
\[ \MB[(U,s), t] = \left\{ \begin{array}{ll}
1, & t|U = s \\
0 & \mbox{otherwise}
\end{array}
\right.
\]
Note that the column $\MB[\_, t]$ of the matrix corresponds to the deterministic model $\dt$.
We can regard a probabilistic model $\{ d_U \}_{U \in \UU}$ as a real vector $\vv$ in the row space of $\MB$, where $\vv[U, s] = d_U(s)$.

\begin{proposition}
\label{linequivglobprop}
The non-contextuality of the probabilistic model represented by the vector $\vv$ is equivalent to the existence of a non-negative solution $\xx \geq \Zero$ for the linear system
\[ \MB \xx = \vv . \]
\end{proposition}
\begin{proof}
For each $U \in \UU$, the sub-vector $\vv_U$ of $\vv$ forms a probability distribution on $\Two^U$, and hence sums to $1$. Since the restriction map $\Two^X \rarr \Two^U$ is surjective, and $\MB \xx = \vv$ implies $(\MB \xx)_U = \vv_U$, it follows that the entries of $\xx$ sum to $1$. Thus $\xx$ defines a probability distribution $\mu$ on $\Two^X$. Moreover, the equation $(\MB \xx)_U = \vv_U$ is equivalent to \[ d_U = \sum_{t \in \Two^X} \mu(t) \delta^t_U . \]
Thus a solution $\xx$ exists if and only if the model can be written as a convex combination as in Theorem~\ref{ncdetth}.
\end{proof}

Thus the set $\NN$ of non-contextual probability models is the convex hull of the set of deterministic models $\dt$, $t \in \Two^X$.
By the fundamental properties of convex polytopes \cite{matou?ek2007understanding,schrijver1998theory,ziegler1995lectures}, $\NN$ is equivalently specified by a finite set of linear inequalities.

To state this more explicitly, we firstly recall the well-known  Fourier-Motzkin elimination procedure \cite{matou?ek2007understanding,schrijver1998theory,ziegler1995lectures}.

\begin{proposition}[Fourier-Motzkin elimination]
\label{FMprop}
If we are given a finite system $I(\xx,\yy)$ of linear inequalities in the variables $\xx$, $\yy$, we can effectively obtain a finite system $J(\yy)$ of inequalities in the variables $\yy$, such that $\vv$ satisfies $J$ if and only if for some $\ww$, $(\ww, \vv)$ satisfies $I$.
Moreover, $J$ is constructed from $I$ using only the field operations, so if $I$ is rational, so is $J$.
\end{proposition}

The size of $J$ is, in the worst case, doubly exponential in the size of $I$.
Nevertheless, Fourier-Motzkin elimination is widely used in computer-assisted verification and polyhedral computation \cite{strichman2002solving,christof1997porta}.

In our case, we begin with the `symbolic' system
\[ \MB \xx = \yy, \qquad \xx \geq \Zero, \qquad \One \cdot \xx = 1 \]
in variables $\xx$, $\yy$. This can be written as
\[  \begin{array}{rlr}
a_{1,j}x_1 + \cdots + a_{N,j}x_N - y_j & \geq 0, & \qquad j = 1, \ldots , D \\
-a_{1,j}x_1 + \cdots + -a_{N,j}x_N + y_j & \geq 0, & \qquad j = 1, \ldots , D \\
x_i & \geq 0, &\qquad i = 1, \ldots , N \\
x_1 + \cdots + x_N & \geq 1 & \\
-x_1 + \cdots + -x_N & \geq -1 &
\end{array}
\]
where the coefficients $a_{i,j}$ come from the incidence matrix $\MB$, and
\[ N := 2^{\card{X}}, \qquad D := \sum_{U \in \UU} 2^{\card{U}}  \]
are the dimensions of $\MB$. Note that, since the system is symbolic, we have to add the constraint that $\xx$ sums to $1$ explicitly.

Writing this system as $I(\xx,\yy)$, by Proposition~\ref{linequivglobprop}, we have
\[ \NN \; = \; \{ \vv \mid \exists \ww. \, I(\ww, \vv) \} . \]
By Proposition~\ref{FMprop}, we can eliminate the variables $\xx$ from this system, producing a system $J$ of inequalities in the variables $\yy$, such that $\vv$ satisfies $J$ if and only if for some $\ww$, $(\ww, \vv)$ satisfies $I$. Thus $\vv$ is in $\NN$ if and only if $\vv$ satisfies $J$.

Thus we obtain the following result.
\begin{proposition}
\label{finratineqprop}
There is a finite set of rational vectors $\rr_1, \ldots , \rr_p$, and rational numbers $r_1, \ldots , r_p$, such that, for all probability models $\vv$:
\[ \vv\in \NN \IFF \forall i = 1, \ldots , p. \; \rr_i \cdot \vv \; \leq \; r_i . \]
\end{proposition}

\subsection{Completeness of logical Bell inequalities}

Suppose we are given a cover $\UU$. A rational inequality over $\UU$ is given by a rational vector $\rr$ and a rational number $r$. A probability model $\vv$ satisfies this inequality if $\rr \cdot \vv \leq r$.
Two inequalities are equivalent if they are satisfied by the same probability models.

\begin{theorem}
\label{ratineqequivlbellth}
A rational inequality is satisfied by all non-contextual models over $\UU$ if and only if it is equivalent to a logical Bell inequality.
\end{theorem}
\begin{proof}
A rational inequality determines an equivalent integer inequality given by an integer vector $\kk$ and an integer $M$, obtained by clearing denominators.

Suppose that we are given an integer vector $\kk$ indexed by $(U, s)$, where $U \in \UU$ and $s \in \Two^U$.
For each $(U, s)$, we define non-negative integers $\kUs$, and formulas $\tUs$ in the variables $U$:
\[ \begin{array}{lcl}
\kUs & = & | \kk[U, s] | \\
\tUs & = & \left\{ \begin{array}{ll}
\vphi_s, & \kk[U,s] \geq 0 \\
\neg \vphi_s, & \kk[U, s] < 0.
\end{array} \right.
\end{array}
\]
Here we use $\vphi_s$ as  defined in~(\ref{phiseq}).

Now suppose we are given a probability model $\vv$.
For each $(U, s)$, we define $\pUs$ to be the probability assigned by $\vv$ to the subset of $\Two^U$ defined by $\tUs$.

We claim that:
\begin{equation}
\label{doteqlog}
\kk \cdot \vv \; = \; \sum_{U,s} \kUs \pUs \; - \; \sum_{\kk[U, s] < 0} \kUs .
\end{equation}
To see this,
for each $(U, s)$ we compare $\kk[U,s]\vv[U,s]$ with $\kUs \pUs$:
\begin{itemize}
\item If $\kk[U,s] \geq 0$, then $\kk[U,s]\vv[U,s] = \kUs \pUs$.
\item If $\kk[U,s] < 0$, we have:
\[ \kk[U,s]\vv[U,s] \; = \;  \kUs((1 - p(\vphi_s)) -1) \; = \; \kUs (\pUs - 1) . \]
\end{itemize}
Collecting terms, we obtain~(\ref{doteqlog}).

We now consider the expression
\begin{equation}
\label{thexp}
\sum_{U,s} \kUs p(\tUs) \; \leq \; K ,
\end{equation}
where $K = M + \sum_{\kk[U, s] < 0} \kUs$.

By~(\ref{doteqlog}), a probability model $\vv$ will satisfy this inequality if and only if $\kk \cdot \vv \leq M$. Thus this inequality is equivalent to the rational inequality we began with.

Now suppose that this inequality is satisfied by all non-contextual models.
Since the coefficients $\kUs$ in~(\ref{thexp}) are non-negative, $K$ must be non-negative.
By Proposition~\ref{kconprop}, the multiset of formulas $\sum_{U, s} \kUs \tUs$ is $K$-consistent, and thus~(\ref{thexp}) is a logical Bell inequality.

Thus every rational inequality satisfied by all non-contextual models is equivalent to a logical Bell inequality.
From Proposition~\ref{kconprop}, we also have the converse:~every logical Bell inequality is satisfied by all non-contextual models.
\end{proof}

Combining Proposition~\ref{finratineqprop} and Theorem~\ref{ratineqequivlbellth}, we obtain the following completeness result.

\begin{theorem}
\label{compth}
The polytope of non-contextual probability models over any cover $\UU$ is determined by a finite set of logical Bell inequalities. Moreover, these inequalities can be obtained effectively from $\UU$.
Thus a probabilistic model over any cover is contextual if and only if it violates one of finitely many logical Bell inequalities.
\end{theorem}
\begin{proof}
By Proposition~\ref{finratineqprop}, given $\UU$ we can effectively obtain a finite set of rational inequalities defining the non-contextual polytope. Using the construction given in the proof of Theorem~\ref{ratineqequivlbellth}, we can effectively transform these rational inequalities into equivalent logical Bell inequalities.
\end{proof}

\section{Logical description of correlation inequalities}

We shall now show that correlation inequalities can also be analyzed logically; in fact, they form a special case of the logical inequalities we have already described.

A probability model $\vv$ determines a vector $\ev = (E_U)_{U \in \UU}$ of expectation values. Here
\[ E_U := (+1) \cdot p(\pU) + (-1)\cdot p(\neg \pU) , \]
where $\pU$ is a formula whose satisfying assignments are those with an even number of $1$ outcomes. Thus we can define
\begin{equation}
\label{pcfdef}
\pU \; := \; \neg \bigoplus_{x \in U} \neg x .
\end{equation}

The set of expectation vectors of non-contextual models is the image under a linear map of the convex polytope of non-contextual probability models, and hence forms a convex polytope $\EP$, with vertices given by the vectors $\et$, $t \in \Two^X$.

Clearly, any probability model $\vv$ whose expectation vector $\ev$ is not in $\EP$ must be contextual. However, the converse is \emph{not} the case. We shall return to this point in Section~\ref{nsigvex}.
Nevertheless, the correlation inequalities have received a great deal of attention in the literature on non-locality, and it is clearly of considerable interest to give a complete characterization.

We shall now give a logical characterization of a complete set of inequalities for the polytope $\EP$ on an arbitrary cover $\UU$.

\begin{theorem}
\label{corrineqth}
For any probability model $\vv$ such that $\ev \not\in \EP$, there is a logical Bell inequality
\begin{equation}
\label{lbcorreq}
\sum_{U \in \UU} k_U p(\tU) \; \leq \; K
\end{equation}
 which is violated by $\vv$, where for each $U$, $\tU$ is either $\pU$ or $\neg \pU$.
\end{theorem}
\begin{proof}
By similar reasoning to that used in the proof of Theorem~\ref{compth}, there is  an integer vector $\kk$ and an integer $M$ such that $\kk \cdot \ew \leq M$ for all non-contextual models $\ww$, and $\kk \cdot \ev > M$.

For each $U$, and any probability model $\ww$, we consider two cases:
\begin{itemize}
\item If $\kk[U] = k_U$ is positive, then we can write $\kk[U]  \ew[U] = k_U(2p(\pU) - 1)$.
\item If $\kk[U] = -k_U$ is negative, we can write
\begin{multline*}
\kk[U]  \ew[U] = -k_U(2 p(\pU) - 1) = k_U(1 - 2p(\pU)) \\
= k_U(2(1- p(\pU)) - 1) = k_U(2p(\neg \pU) -1) . 
\end{multline*}
\end{itemize}
Rearranging terms, we have
\[ \kk \cdot \ew \; = \; \sum_{U \in \UU} 2k_U p(\tU)  - P \]
where each $\tU$ is either $\pU$ or $\neg \pU$, and $P$ is a positive integer. Hence the inequality $\kk \cdot \ew \leq M$ is equivalent to $\ww$ satisfying the inequality
\begin{equation}
\label{ctheq}
\sum_{U \in \UU} 2 k_U p(\tU) \; \leq \; K
\end{equation}
where $K = M + P$. By Proposition~\ref{kconprop}, the fact that all non-contextual models $\ww$ satisfy $\kk \cdot \ew \leq M$ implies that~(\ref{ctheq}) is a logical Bell inequality.
Since $\kk \cdot \ev > M$, $\vv$ violates this inequality.
\end{proof}

Note that, for any vector $\eta \in \EP$, $\eta = \ew$ for some non-contextual model $\ww$, and $\ww$ satisfies all the logical Bell inequalities.

It is also possible to reverse the procedure described in Theorem~\ref{corrineqth}, to obtain a complete set of inequalities directly applicable to expectation vectors.

Given a logical Bell inequality of the form
\[ \sum_{U \in \UU} 2 k_U p(\tU) \; \leq K \]
where for each $U$, $\tU$ is either $\pU$ or $\neg \pU$, we can form the inequality
\[  \sum_{U \in \UU} l_U E_U \; \leq \; M \]
where $M = K - \sum_{U \in \UU} k_U$, and
\[ l_U = \left\{ \begin{array}{rl}
k_U, & \tU = \pU \\
-k_U, & \tU = \neg \pU .
\end{array}
\right.
\]
We call this class of inequalities on expectation vectors the \emph{logical correlation inequalities}.

As an immediate consequence of Theorem~\ref{corrineqth}, we have.
\begin{theorem}
\label{ecorrineqth}
An expectation vector $\eta$ is in $\EP$ if and only if it satisfies all the logical correlation inequalities.
\end{theorem}

\subsection{Example}
We consider the following correlation inequality for the $(3, 2, 2)$ case given by Werner and Wolf in \cite{werner2001all}:
\[ 1/4 \sum_{i=1}^8 E_i  \; - \;  E_8 \; \leq \; 1 \qquad  (A2). \]
Here $E_i$, for $i = 1, \ldots , 8$, is the expectation value for the combination of measurements whose value, written as a binary string,  is $i-1$.

If we write this more explicitly, and clear the denominator of the scaling factor $1/4$, we obtain:
\[ \sum_{i=1}^7 E_i \, - \, 3E_8 \; \leq \; 4 . \]
If we now convert this to the form~(\ref{bellineqexp}), following the procedure given in the proof of Theorem~\ref{corrineqth}, we obtain the following inequality:
\[ \sum_{i=1}^7 p(\pcf_i) \, + \, 3 p(\neg \pcf_8) \; \leq \; 7. \]
We can see that the multiset
\[  \sum_{i=1}^7 1\pcf_i \, + \, 3 (\neg \pcf_8) \]
is $7$-consistent. In fact, $\neg \pcf_8$, together with any 5 of the formulas $\pcf_1, \ldots , \pcf_7$, is inconsistent.

\subsection{Example}
\label{nsigvex}

There are a number of extremal vertices of the no-signalling polytope in the $(3,2,2)$ case, as listed in \cite{pironio2011extremal}, which satisfy all the correlation inequalities from \cite{werner2001all} \footnote{The first author is grateful to Matty Hoban for  bringing these examples to his attention.}.

We shall examine one of these in detail. This is the vertex numbered 4 in the listing in \cite{pironio2011extremal}.

We shall label the measurements as $a$, $a'$ for site 1; $b$, $b'$ for site 2; and $c$, $c'$ for site 3.
The support of the model for each measurement combination $m$ can be represented by a formula $\vphi_m$; since the distribution on each row is uniform on the support, this completely specifies the model.

We recall the definition of $\pU$ from~(\ref{pcfdef}).
The formulas for the support of the model are defined as follows:
\begin{multline*}
\vphi_{abc} = \vphi_{abc'} = \pcf_{ab}; \quad \vphi_{ab'c'} = \vphi_{a'b'c'} = \pcf_{b'c'} \\
\vphi_{a'bc} = \vphi_{a'b'c} = \pcf_{a'c}; 
\;\;\; \vphi_{ab'c} = \pcf_{ab'c}; \;\;\; \vphi_{a'bc'} = \neg \pcf_{a'bc'} .
\end{multline*}
Combining these, we obtain the following multiset of formulas:
\[ 2 \pcf_{ab} + 2 \pcf_{b'c'} + 2 \pcf_{a'c} + \pcf_{ab'c} + \neg \pcf_{a'bc'} . \]
Since $\pcf_{ab}$ is equivalent to $a \leftrightarrow b$, in the presence of the first three formulas $\pcf_{ab'c}$ is equivalent to $\pcf_{bc'c}$, and $\neg \pcf_{a'bc'}$ is equivalent to $\neg \pcf_{cbc'}$.
Since $\pU$ is independent of the order in which the elements of $U$ are listed, we obtain a contradiction. In fact, this multiset of formulas is $7$-consistent, so the model achieves a maximal violation of the logical Bell inequality
\[ 2 p(\pcf_{ab}) + 2 p(\pcf_{b'c'}) + 2 p(\pcf_{a'c}) + p(\pcf_{ab'c}) + p(\neg \pcf_{a'bc'}) \; \leq \; 7. \]
This yields a concrete example of a no-signalling model which satisfies all the correlation inequalities, while maximally violating the canonical logical Bell inequality arising from its support.

\section{Multiple Outcomes}
\label{moutsec}

Thus far we have focussed exclusively on dichotomic measurements, which are particularly convenient for connecting to logic. However, the general format of measurement covers easily allows the results to be extended to measurements with multiple outcomes.

For example, we consider the case of $(n, k, 2^p)$ Bell scenarios: $n$ sites, $k$ measurements per site, and $2^p$ outcomes per measurement. This corresponds to the following situation in our setting. We have a set $X$ with $nkp$ elements $\{ \mijl \}$, where $i = 1, \ldots , n$, $j = 1, \ldots , k$, and $l = 1, \ldots , p$. We write
\[ \Xij := \{ \mijl \mid l = 1, \ldots , p \}, \qquad X_i := \bigcup_{j=1}^k \Xij . \]
The cover $\UU$ comprises all those subsets $U$ of $X$ such that, for all $i = 1, \ldots , n$, for some $j$, $U \cap X_i = \Xij$. The idea is that $X_i$ is the set of measurements which can be performed at site $i$. There are $k$ choices available at each site between sets $\Xij$ of $p$ dichotomic measurements each. Because these measurements are compatible, they can be performed together, resulting in a measurement with $2^p$ possible outcomes. An overall choice of measurements consists of selecting one such compatible family for each site.

All our results apply directly to this situation, which is itself a very special case of the general notion of cover.
Thus from Theorem~\ref{ecorrineqth}, we have an explicit description of a complete set of correlation inequalities characterizing the $(n, k, 2^p)$ Bell scenarios.


\section{Final Remarks}

For further directions, it would be of particular interest to see how much of the present approach could be lifted to the quantum set, and the Tsirelson inequality \cite{tsirelson1980}.

As regards related work,
the form of expressions we have used for the logical inequalities correspond to basic weight formulas in the logic for reasoning about probabilities studied in \cite{fagin1990logic}, following \cite{nilsson1986probabilistic}, which was motivated by applications in Artifical Intelligence.

The correlation polytopes of Pitowsky \cite{pitowsky1994george}, which have a lineage going back to Boole's `conditions of possible experience' \cite{boole1862theory}, should also be mentioned.
Although this line of thought is certainly in a kindred spirit, Boole's conditions are arithmetical in nature; while the central theme of the present paper is that complete sets of Bell inequalities can be defined in terms of purely logical consistency conditions.

The notion of $K$-consistency is closely related to the well-known \textsf{MAX-SAT} problem in computational complexity \cite{papadimitriou2003computational}.
This asks for the maximum number of clauses in a given set which are satisfiable.

\begin{acknowledgments}
The first author thanks Harvey Brown, Matty Hoban, Shane Mansfield and Rui Soares Barbosa for a number of valuable discussions.
\end{acknowledgments}

\bibliography{bdbib}

\begin{thebibliography}{43}%
\makeatletter
\providecommand \@ifxundefined [1]{%
 \@ifx{#1\undefined}
}%
\providecommand \@ifnum [1]{%
 \ifnum #1\expandafter \@firstoftwo
 \else \expandafter \@secondoftwo
 \fi
}%
\providecommand \@ifx [1]{%
 \ifx #1\expandafter \@firstoftwo
 \else \expandafter \@secondoftwo
 \fi
}%
\providecommand \natexlab [1]{#1}%
\providecommand \enquote  [1]{``#1''}%
\providecommand \bibnamefont  [1]{#1}%
\providecommand \bibfnamefont [1]{#1}%
\providecommand \citenamefont [1]{#1}%
\providecommand \href@noop [0]{\@secondoftwo}%
\providecommand \href [0]{\begingroup \@sanitize@url \@href}%
\providecommand \@href[1]{\@@startlink{#1}\@@href}%
\providecommand \@@href[1]{\endgroup#1\@@endlink}%
\providecommand \@sanitize@url [0]{\catcode `\\12\catcode `\$12\catcode
  `\&12\catcode `\#12\catcode `\^12\catcode `\_12\catcode `\%12\relax}%
\providecommand \@@startlink[1]{}%
\providecommand \@@endlink[0]{}%
\providecommand \url  [0]{\begingroup\@sanitize@url \@url }%
\providecommand \@url [1]{\endgroup\@href {#1}{\urlprefix }}%
\providecommand \urlprefix  [0]{URL }%
\providecommand \Eprint [0]{\href }%
\@ifxundefined \urlstyle {%
  \providecommand \doi  [0]{\begingroup \@sanitize@url \@doi}%
  \providecommand \@doi [1]{\endgroup \@@startlink {\doibase
  #1}doi:\discretionary {}{}{}#1\@@endlink }%
}{%
  \providecommand \doi  [0]{doi:\discretionary{}{}{}\begingroup
  \urlstyle{rm}\Url }%
}%
\providecommand \doibase [0]{http://dx.doi.org/}%
\providecommand \Doi [0]{\begingroup \@sanitize@url \@Doi }%
\providecommand \@Doi  [1]{\endgroup\@@startlink{\doibase#1}\@@Doi}%
\providecommand \@@Doi [1]{#1\@@endlink}%
\providecommand \selectlanguage [0]{\@gobble}%
\providecommand \bibinfo  [0]{\@secondoftwo}%
\providecommand \bibfield  [0]{\@secondoftwo}%
\providecommand \translation [1]{[#1]}%
\providecommand \BibitemOpen [0]{}%
\providecommand \bibitemStop [0]{}%
\providecommand \bibitemNoStop [0]{.\EOS\space}%
\providecommand \EOS [0]{\spacefactor3000\relax}%
\providecommand \BibitemShut  [1]{\csname bibitem#1\endcsname}%
\bibitem [{\citenamefont {Bell}(1964)}]{bell1964einstein}%
  \BibitemOpen
  \bibfield  {author} {\bibinfo {author} {\bibfnamefont {J.}~\bibnamefont
  {Bell}},\ }\href@noop {} {\bibfield  {journal} {\bibinfo  {journal}
  {Physics},\ }\textbf {\bibinfo {volume} {1}},\ \bibinfo {pages} {195}
  (\bibinfo {year} {1964})}\BibitemShut {NoStop}%
\bibitem [{\citenamefont {Clauser}\ \emph {et~al.}(1969)\citenamefont
  {Clauser}, \citenamefont {Horne}, \citenamefont {Shimony},\ and\
  \citenamefont {Holt}}]{clauser1969proposed}%
  \BibitemOpen
  \bibfield  {author} {\bibinfo {author} {\bibfnamefont {J.}~\bibnamefont
  {Clauser}}, \bibinfo {author} {\bibfnamefont {M.}~\bibnamefont {Horne}},
  \bibinfo {author} {\bibfnamefont {A.}~\bibnamefont {Shimony}}, \ and\
  \bibinfo {author} {\bibfnamefont {R.}~\bibnamefont {Holt}},\ }\href@noop {}
  {\bibfield  {journal} {\bibinfo  {journal} {Physical Review Letters},\
  }\textbf {\bibinfo {volume} {23}},\ \bibinfo {pages} {880} (\bibinfo {year}
  {1969})}\BibitemShut {NoStop}%
\bibitem [{\citenamefont {Ekert}(1991)}]{ekert1991quantum}%
  \BibitemOpen
  \bibfield  {author} {\bibinfo {author} {\bibfnamefont {A.}~\bibnamefont
  {Ekert}},\ }\href@noop {} {\bibfield  {journal} {\bibinfo  {journal}
  {Physical review letters},\ }\textbf {\bibinfo {volume} {67}},\ \bibinfo
  {pages} {661} (\bibinfo {year} {1991})}\BibitemShut {NoStop}%
\bibitem [{\citenamefont {Barrett}\ \emph {et~al.}(2005)\citenamefont
  {Barrett}, \citenamefont {Hardy},\ and\ \citenamefont
  {Kent}}]{barrett2005no}%
  \BibitemOpen
  \bibfield  {author} {\bibinfo {author} {\bibfnamefont {J.}~\bibnamefont
  {Barrett}}, \bibinfo {author} {\bibfnamefont {L.}~\bibnamefont {Hardy}}, \
  and\ \bibinfo {author} {\bibfnamefont {A.}~\bibnamefont {Kent}},\ }\href@noop
  {} {\bibfield  {journal} {\bibinfo  {journal} {Physical review letters},\
  }\textbf {\bibinfo {volume} {95}},\ \bibinfo {pages} {10503} (\bibinfo {year}
  {2005})}\BibitemShut {NoStop}%
\bibitem [{\citenamefont {Acin}\ \emph {et~al.}(2006)\citenamefont {Acin},
  \citenamefont {Gisin},\ and\ \citenamefont {Masanes}}]{acin2006bell}%
  \BibitemOpen
  \bibfield  {author} {\bibinfo {author} {\bibfnamefont {A.}~\bibnamefont
  {Acin}}, \bibinfo {author} {\bibfnamefont {N.}~\bibnamefont {Gisin}}, \ and\
  \bibinfo {author} {\bibfnamefont {L.}~\bibnamefont {Masanes}},\ }\href@noop
  {} {\bibfield  {journal} {\bibinfo  {journal} {Physical review letters},\
  }\textbf {\bibinfo {volume} {97}},\ \bibinfo {pages} {120405} (\bibinfo
  {year} {2006})}\BibitemShut {NoStop}%
\bibitem [{\citenamefont {Brukner}\ \emph {et~al.}(2004)\citenamefont
  {Brukner}, \citenamefont {{\.Z}ukowski}, \citenamefont {Pan},\ and\
  \citenamefont {Zeilinger}}]{brukner2004bell}%
  \BibitemOpen
  \bibfield  {author} {\bibinfo {author} {\bibfnamefont {{\v{C}}.}~\bibnamefont
  {Brukner}}, \bibinfo {author} {\bibfnamefont {M.}~\bibnamefont
  {{\.Z}ukowski}}, \bibinfo {author} {\bibfnamefont {J.}~\bibnamefont {Pan}}, \
  and\ \bibinfo {author} {\bibfnamefont {A.}~\bibnamefont {Zeilinger}},\
  }\href@noop {} {\bibfield  {journal} {\bibinfo  {journal} {Physical review
  letters},\ }\textbf {\bibinfo {volume} {92}},\ \bibinfo {pages} {127901}
  (\bibinfo {year} {2004})}\BibitemShut {NoStop}%
\bibitem [{\citenamefont {Terhal}(2002)}]{terhal2002detecting}%
  \BibitemOpen
  \bibfield  {author} {\bibinfo {author} {\bibfnamefont {B.}~\bibnamefont
  {Terhal}},\ }\href@noop {} {\bibfield  {journal} {\bibinfo  {journal}
  {Theoretical computer science},\ }\textbf {\bibinfo {volume} {287}},\
  \bibinfo {pages} {313} (\bibinfo {year} {2002})}\BibitemShut {NoStop}%
\bibitem [{\citenamefont {Bartosik}\ \emph {et~al.}(2009)\citenamefont
  {Bartosik}, \citenamefont {Klepp}, \citenamefont {Schmitzer}, \citenamefont
  {Sponar}, \citenamefont {Cabello}, \citenamefont {Rauch},\ and\ \citenamefont
  {Hasegawa}}]{bartosik2009experimental}%
  \BibitemOpen
  \bibfield  {author} {\bibinfo {author} {\bibfnamefont {H.}~\bibnamefont
  {Bartosik}}, \bibinfo {author} {\bibfnamefont {J.}~\bibnamefont {Klepp}},
  \bibinfo {author} {\bibfnamefont {C.}~\bibnamefont {Schmitzer}}, \bibinfo
  {author} {\bibfnamefont {S.}~\bibnamefont {Sponar}}, \bibinfo {author}
  {\bibfnamefont {A.}~\bibnamefont {Cabello}}, \bibinfo {author} {\bibfnamefont
  {H.}~\bibnamefont {Rauch}}, \ and\ \bibinfo {author} {\bibfnamefont
  {Y.}~\bibnamefont {Hasegawa}},\ }\href@noop {} {\bibfield  {journal}
  {\bibinfo  {journal} {Physical Review Letters},\ }\textbf {\bibinfo {volume}
  {103}},\ \bibinfo {pages} {40403} (\bibinfo {year} {2009})}\BibitemShut
  {NoStop}%
\bibitem [{\citenamefont {Kirchmair}\ \emph {et~al.}(2009)\citenamefont
  {Kirchmair}, \citenamefont {Z{\"a}hringer}, \citenamefont {Gerritsma},
  \citenamefont {Kleinmann}, \citenamefont {G{\"u}hne}, \citenamefont
  {Cabello}, \citenamefont {Blatt},\ and\ \citenamefont
  {Roos}}]{kirchmair2009state}%
  \BibitemOpen
  \bibfield  {author} {\bibinfo {author} {\bibfnamefont {G.}~\bibnamefont
  {Kirchmair}}, \bibinfo {author} {\bibfnamefont {F.}~\bibnamefont
  {Z{\"a}hringer}}, \bibinfo {author} {\bibfnamefont {R.}~\bibnamefont
  {Gerritsma}}, \bibinfo {author} {\bibfnamefont {M.}~\bibnamefont
  {Kleinmann}}, \bibinfo {author} {\bibfnamefont {O.}~\bibnamefont
  {G{\"u}hne}}, \bibinfo {author} {\bibfnamefont {A.}~\bibnamefont {Cabello}},
  \bibinfo {author} {\bibfnamefont {R.}~\bibnamefont {Blatt}}, \ and\ \bibinfo
  {author} {\bibfnamefont {C.}~\bibnamefont {Roos}},\ }\href@noop {} {\bibfield
   {journal} {\bibinfo  {journal} {Nature},\ }\textbf {\bibinfo {volume}
  {460}},\ \bibinfo {pages} {494} (\bibinfo {year} {2009})}\BibitemShut
  {NoStop}%
\bibitem [{\citenamefont {Greenberger}\ \emph {et~al.}(1990)\citenamefont
  {Greenberger}, \citenamefont {Horne}, \citenamefont {Shimony},\ and\
  \citenamefont {Zeilinger}}]{greenberger1990bell}%
  \BibitemOpen
  \bibfield  {author} {\bibinfo {author} {\bibfnamefont {D.}~\bibnamefont
  {Greenberger}}, \bibinfo {author} {\bibfnamefont {M.}~\bibnamefont {Horne}},
  \bibinfo {author} {\bibfnamefont {A.}~\bibnamefont {Shimony}}, \ and\
  \bibinfo {author} {\bibfnamefont {A.}~\bibnamefont {Zeilinger}},\ }\href@noop
  {} {\bibfield  {journal} {\bibinfo  {journal} {American Journal of Physics},\
  }\textbf {\bibinfo {volume} {58}},\ \bibinfo {pages} {1131} (\bibinfo {year}
  {1990})}\BibitemShut {NoStop}%
\bibitem [{\citenamefont {Hardy}(1992)}]{hardy1992quantum}%
  \BibitemOpen
  \bibfield  {author} {\bibinfo {author} {\bibfnamefont {L.}~\bibnamefont
  {Hardy}},\ }\href@noop {} {\bibfield  {journal} {\bibinfo  {journal}
  {Physical Review Letters},\ }\textbf {\bibinfo {volume} {68}},\ \bibinfo
  {pages} {2981} (\bibinfo {year} {1992})}\BibitemShut {NoStop}%
\bibitem [{\citenamefont {Zimba}\ and\ \citenamefont
  {Penrose}(1993)}]{zimba1993bell}%
  \BibitemOpen
  \bibfield  {author} {\bibinfo {author} {\bibfnamefont {J.}~\bibnamefont
  {Zimba}}\ and\ \bibinfo {author} {\bibfnamefont {R.}~\bibnamefont
  {Penrose}},\ }\href@noop {} {\bibfield  {journal} {\bibinfo  {journal}
  {Studies in History and Philosophy of Science Part A},\ }\textbf {\bibinfo
  {volume} {24}},\ \bibinfo {pages} {697} (\bibinfo {year} {1993})}\BibitemShut
  {NoStop}%
\bibitem [{\citenamefont {Abramsky}\ and\ \citenamefont
  {Brandenburger}(2011)}]{abramsky2011unified}%
  \BibitemOpen
  \bibfield  {author} {\bibinfo {author} {\bibfnamefont {S.}~\bibnamefont
  {Abramsky}}\ and\ \bibinfo {author} {\bibfnamefont {A.}~\bibnamefont
  {Brandenburger}},\ }\href@noop {} {\bibfield  {journal} {\bibinfo  {journal}
  {New Journal of Physics},\ }\textbf {\bibinfo {volume} {13(2011)}},\ \bibinfo
  {pages} {113036} (\bibinfo {year} {2011})}\BibitemShut {NoStop}%
\bibitem [{\citenamefont {Hardy}(1993)}]{hardy1993nonlocality}%
  \BibitemOpen
  \bibfield  {author} {\bibinfo {author} {\bibfnamefont {L.}~\bibnamefont
  {Hardy}},\ }\href@noop {} {\bibfield  {journal} {\bibinfo  {journal}
  {Physical Review Letters},\ }\textbf {\bibinfo {volume} {71}},\ \bibinfo
  {pages} {1665} (\bibinfo {year} {1993})}\BibitemShut {NoStop}%
\bibitem [{\citenamefont {Hardy}(1991)}]{hardy1991new}%
  \BibitemOpen
  \bibfield  {author} {\bibinfo {author} {\bibfnamefont {L.}~\bibnamefont
  {Hardy}},\ }\href@noop {} {\bibfield  {journal} {\bibinfo  {journal} {Physics
  Letters A},\ }\textbf {\bibinfo {volume} {161}},\ \bibinfo {pages} {21}
  (\bibinfo {year} {1991})}\BibitemShut {NoStop}%
\bibitem [{\citenamefont {Braunstein}\ and\ \citenamefont
  {Caves}(1990)}]{braunstein1990wringing}%
  \BibitemOpen
  \bibfield  {author} {\bibinfo {author} {\bibfnamefont {S.}~\bibnamefont
  {Braunstein}}\ and\ \bibinfo {author} {\bibfnamefont {C.}~\bibnamefont
  {Caves}},\ }\href@noop {} {\bibfield  {journal} {\bibinfo  {journal} {Annals
  of Physics},\ }\textbf {\bibinfo {volume} {202}},\ \bibinfo {pages} {22}
  (\bibinfo {year} {1990})}\BibitemShut {NoStop}%
\bibitem [{Note1()}]{Note1}%
  \BibitemOpen
  \bibinfo {note} {Since each $p_i$ is a probability, its maximum value is $1$,
  so the `algebraic maximum' of the sum $\DOTSB \sum@ \slimits@ _i p_i$ is
  $N$.}\BibitemShut {Stop}%
\bibitem [{\citenamefont {Abramsky}\ \emph {et~al.}(2012)\citenamefont
  {Abramsky}, \citenamefont {Mansfield},\ and\ \citenamefont
  {Barbosa}}]{abramsky2011cohomology}%
  \BibitemOpen
  \bibfield  {author} {\bibinfo {author} {\bibfnamefont {S.}~\bibnamefont
  {Abramsky}}, \bibinfo {author} {\bibfnamefont {S.}~\bibnamefont {Mansfield}},
  \ and\ \bibinfo {author} {\bibfnamefont {R.}~\bibnamefont {Barbosa}},\ }in\
  \href@noop {} {\emph {\bibinfo {booktitle} {Proceedings of QPL 2011}}}\
  (\bibinfo  {publisher} {EPTCS},\ \bibinfo {year} {2012})\ \bibinfo {note}
  {available as arXiv:1111.3620v1}\BibitemShut {NoStop}%
\bibitem [{\citenamefont {Fine}(1982)}]{fine1982hidden}%
  \BibitemOpen
  \bibfield  {author} {\bibinfo {author} {\bibfnamefont {A.}~\bibnamefont
  {Fine}},\ }\href@noop {} {\bibfield  {journal} {\bibinfo  {journal} {Physical
  Review Letters},\ }\textbf {\bibinfo {volume} {48}},\ \bibinfo {pages} {291}
  (\bibinfo {year} {1982})}\BibitemShut {NoStop}%
\bibitem [{Note2()}]{Note2}%
  \BibitemOpen
  \bibinfo {note} {We shall eventually consider a much more general form of
  such models introduced in \cite {abramsky2011unified}, which allow a uniform
  treatment of contextuality, including Kochen-Specker configurations
  etc.}\BibitemShut {Stop}%
\bibitem [{Note3()}]{Note3}%
  \BibitemOpen
  \bibinfo {note} {The support of a probability model can be given the
  structure of a \protect \textbf {presheaf}, in such a way that the above
  definition corresponds to the usual notion of global section. See \cite
  {abramsky2011unified} for details.}\BibitemShut {Stop}%
\bibitem [{\citenamefont {Greenberger}\ \emph {et~al.}(1989)\citenamefont
  {Greenberger}, \citenamefont {Horne},\ and\ \citenamefont
  {Zeilinger}}]{greenberger1989going}%
  \BibitemOpen
  \bibfield  {author} {\bibinfo {author} {\bibfnamefont {D.}~\bibnamefont
  {Greenberger}}, \bibinfo {author} {\bibfnamefont {M.}~\bibnamefont {Horne}},
  \ and\ \bibinfo {author} {\bibfnamefont {A.}~\bibnamefont {Zeilinger}},\ }in\
  \href@noop {} {\emph {\bibinfo {booktitle} {Bell's Theorem, Quantum Theory,
  and Conceptions of the Universe}}},\ \bibinfo {editor} {edited by\ \bibinfo
  {editor} {\bibfnamefont {M.}~\bibnamefont {Kafatos}}}\ (\bibinfo  {publisher}
  {Kluwer},\ \bibinfo {year} {1989})\ pp.\ \bibinfo {pages}
  {69--72}\BibitemShut {NoStop}%
\bibitem [{\citenamefont {Mermin}(1990){\natexlab{a}}}]{mermin1990quantum}%
  \BibitemOpen
  \bibfield  {author} {\bibinfo {author} {\bibfnamefont {N.}~\bibnamefont
  {Mermin}},\ }\href@noop {} {\bibfield  {journal} {\bibinfo  {journal} {Am. J.
  Phys},\ }\textbf {\bibinfo {volume} {58}},\ \bibinfo {pages} {731} (\bibinfo
  {year} {1990}{\natexlab{a}})}\BibitemShut {NoStop}%
\bibitem [{\citenamefont {Popescu}\ and\ \citenamefont
  {Rohrlich}(1994)}]{popescu1994quantum}%
  \BibitemOpen
  \bibfield  {author} {\bibinfo {author} {\bibfnamefont {S.}~\bibnamefont
  {Popescu}}\ and\ \bibinfo {author} {\bibfnamefont {D.}~\bibnamefont
  {Rohrlich}},\ }\href@noop {} {\bibfield  {journal} {\bibinfo  {journal}
  {Foundations of Physics},\ }\textbf {\bibinfo {volume} {24}},\ \bibinfo
  {pages} {379} (\bibinfo {year} {1994})}\BibitemShut {NoStop}%
\bibitem [{Note4()}]{Note4}%
  \BibitemOpen
  \bibinfo {note} {For further details on this perspective on the
  Kochen-Specker theorem, see \cite {abramsky2011unified}.}\BibitemShut {Stop}%
\bibitem [{\citenamefont {Cabello}\ \emph {et~al.}(1996)\citenamefont
  {Cabello}, \citenamefont {Estebaranz},\ and\ \citenamefont
  {Garc{\'\i}a-Alcaine}}]{cabello1996bell}%
  \BibitemOpen
  \bibfield  {author} {\bibinfo {author} {\bibfnamefont {A.}~\bibnamefont
  {Cabello}}, \bibinfo {author} {\bibfnamefont {J.}~\bibnamefont {Estebaranz}},
  \ and\ \bibinfo {author} {\bibfnamefont {G.}~\bibnamefont
  {Garc{\'\i}a-Alcaine}},\ }\href@noop {} {\bibfield  {journal} {\bibinfo
  {journal} {Physics Letters A},\ }\textbf {\bibinfo {volume} {212}},\ \bibinfo
  {pages} {183} (\bibinfo {year} {1996})}\BibitemShut {NoStop}%
\bibitem [{\citenamefont {Peres}(1990)}]{peres1990incompatible}%
  \BibitemOpen
  \bibfield  {author} {\bibinfo {author} {\bibfnamefont {A.}~\bibnamefont
  {Peres}},\ }\href@noop {} {\bibfield  {journal} {\bibinfo  {journal} {Physics
  Letters A},\ }\textbf {\bibinfo {volume} {151}},\ \bibinfo {pages} {107}
  (\bibinfo {year} {1990})}\BibitemShut {NoStop}%
\bibitem [{\citenamefont {Mermin}(1990){\natexlab{b}}}]{mermin1990simple}%
  \BibitemOpen
  \bibfield  {author} {\bibinfo {author} {\bibfnamefont {N.}~\bibnamefont
  {Mermin}},\ }\href@noop {} {\bibfield  {journal} {\bibinfo  {journal}
  {Physical Review Letters},\ }\textbf {\bibinfo {volume} {65}},\ \bibinfo
  {pages} {3373} (\bibinfo {year} {1990}{\natexlab{b}})}\BibitemShut {NoStop}%
\bibitem [{Note5()}]{Note5}%
  \BibitemOpen
  \bibinfo {note} {In this note, following \cite {abramsky2011unified}, we only
  consider finite sets of measurements, and thus discrete probability
  distributions suffice.}\BibitemShut {Stop}%
\bibitem [{\citenamefont {Matou{\v{s}}ek}\ and\ \citenamefont
  {G{\"a}rtner}(2007)}]{matou?ek2007understanding}%
  \BibitemOpen
  \bibfield  {author} {\bibinfo {author} {\bibfnamefont {J.}~\bibnamefont
  {Matou{\v{s}}ek}}\ and\ \bibinfo {author} {\bibfnamefont {B.}~\bibnamefont
  {G{\"a}rtner}},\ }\href@noop {} {\emph {\bibinfo {title} {Understanding and
  using linear programming}}}\ (\bibinfo  {publisher} {Springer Verlag},\
  \bibinfo {year} {2007})\BibitemShut {NoStop}%
\bibitem [{\citenamefont {Schrijver}(1998)}]{schrijver1998theory}%
  \BibitemOpen
  \bibfield  {author} {\bibinfo {author} {\bibfnamefont {A.}~\bibnamefont
  {Schrijver}},\ }\href@noop {} {\emph {\bibinfo {title} {Theory of linear and
  integer programming}}}\ (\bibinfo  {publisher} {John Wiley \& Sons Inc},\
  \bibinfo {year} {1998})\BibitemShut {NoStop}%
\bibitem [{\citenamefont {Ziegler}(1995)}]{ziegler1995lectures}%
  \BibitemOpen
  \bibfield  {author} {\bibinfo {author} {\bibfnamefont {G.}~\bibnamefont
  {Ziegler}},\ }\href@noop {} {\emph {\bibinfo {title} {Lectures on
  polytopes}}},\ Vol.\ \bibinfo {volume} {152}\ (\bibinfo  {publisher}
  {Springer},\ \bibinfo {year} {1995})\BibitemShut {NoStop}%
\bibitem [{\citenamefont {Strichman}(2002)}]{strichman2002solving}%
  \BibitemOpen
  \bibfield  {author} {\bibinfo {author} {\bibfnamefont {O.}~\bibnamefont
  {Strichman}},\ }in\ \href@noop {} {\emph {\bibinfo {booktitle} {Formal
  Methods in Computer-Aided Design}}}\ (\bibinfo {organization} {Springer},\
  \bibinfo {year} {2002})\ pp.\ \bibinfo {pages} {160--170}\BibitemShut
  {NoStop}%
\bibitem [{\citenamefont {Christof}\ \emph {et~al.}(1997)\citenamefont
  {Christof}, \citenamefont {L{\"o}bel},\ and\ \citenamefont
  {Stoer}}]{christof1997porta}%
  \BibitemOpen
  \bibfield  {author} {\bibinfo {author} {\bibfnamefont {T.}~\bibnamefont
  {Christof}}, \bibinfo {author} {\bibfnamefont {A.}~\bibnamefont {L{\"o}bel}},
  \ and\ \bibinfo {author} {\bibfnamefont {M.}~\bibnamefont {Stoer}},\
  }\href@noop {} {\bibfield  {journal} {\bibinfo  {journal} {Publicly available
  via ftp://ftp. zib. de/pub/Packages/mathprog/polyth/porta}} (\bibinfo {year}
  {1997})}\BibitemShut {NoStop}%
\bibitem [{\citenamefont {Werner}\ and\ \citenamefont
  {Wolf}(2001)}]{werner2001all}%
  \BibitemOpen
  \bibfield  {author} {\bibinfo {author} {\bibfnamefont {R.}~\bibnamefont
  {Werner}}\ and\ \bibinfo {author} {\bibfnamefont {M.}~\bibnamefont {Wolf}},\
  }\href@noop {} {\bibfield  {journal} {\bibinfo  {journal} {Physical Review
  A},\ }\textbf {\bibinfo {volume} {64}},\ \bibinfo {pages} {32112} (\bibinfo
  {year} {2001})}\BibitemShut {NoStop}%
\bibitem [{\citenamefont {Pironio}\ \emph {et~al.}(2011)\citenamefont
  {Pironio}, \citenamefont {Bancal},\ and\ \citenamefont
  {Scarani}}]{pironio2011extremal}%
  \BibitemOpen
  \bibfield  {author} {\bibinfo {author} {\bibfnamefont {S.}~\bibnamefont
  {Pironio}}, \bibinfo {author} {\bibfnamefont {J.}~\bibnamefont {Bancal}}, \
  and\ \bibinfo {author} {\bibfnamefont {V.}~\bibnamefont {Scarani}},\
  }\href@noop {} {\bibfield  {journal} {\bibinfo  {journal} {Journal of Physics
  A: Mathematical and Theoretical},\ }\textbf {\bibinfo {volume} {44}},\
  \bibinfo {pages} {065303} (\bibinfo {year} {2011})}\BibitemShut {NoStop}%
\bibitem [{Note6()}]{Note6}%
  \BibitemOpen
  \bibinfo {note} {The first author is grateful to Matty Hoban for bringing
  these examples to his attention.}\BibitemShut {Stop}%
\bibitem [{\citenamefont {Tsirelson}(1980)}]{tsirelson1980}%
  \BibitemOpen
  \bibfield  {author} {\bibinfo {author} {\bibfnamefont {B.}~\bibnamefont
  {Tsirelson}},\ }\href@noop {} {\bibfield  {journal} {\bibinfo  {journal}
  {Letters in Mathematical Physics},\ }\textbf {\bibinfo {volume} {4}},\
  \bibinfo {pages} {93} (\bibinfo {year} {1980})}\BibitemShut {NoStop}%
\bibitem [{\citenamefont {Fagin}\ \emph {et~al.}(1990)\citenamefont {Fagin},
  \citenamefont {Halpern},\ and\ \citenamefont {Megiddo}}]{fagin1990logic}%
  \BibitemOpen
  \bibfield  {author} {\bibinfo {author} {\bibfnamefont {R.}~\bibnamefont
  {Fagin}}, \bibinfo {author} {\bibfnamefont {J.}~\bibnamefont {Halpern}}, \
  and\ \bibinfo {author} {\bibfnamefont {N.}~\bibnamefont {Megiddo}},\
  }\href@noop {} {\bibfield  {journal} {\bibinfo  {journal} {Information and
  computation},\ }\textbf {\bibinfo {volume} {87}},\ \bibinfo {pages} {78}
  (\bibinfo {year} {1990})}\BibitemShut {NoStop}%
\bibitem [{\citenamefont {Nilsson}(1986)}]{nilsson1986probabilistic}%
  \BibitemOpen
  \bibfield  {author} {\bibinfo {author} {\bibfnamefont {N.}~\bibnamefont
  {Nilsson}},\ }\href@noop {} {\bibfield  {journal} {\bibinfo  {journal}
  {Artificial intelligence},\ }\textbf {\bibinfo {volume} {28}},\ \bibinfo
  {pages} {71} (\bibinfo {year} {1986})}\BibitemShut {NoStop}%
\bibitem [{\citenamefont {Pitowsky}(1994)}]{pitowsky1994george}%
  \BibitemOpen
  \bibfield  {author} {\bibinfo {author} {\bibfnamefont {I.}~\bibnamefont
  {Pitowsky}},\ }\href@noop {} {\bibfield  {journal} {\bibinfo  {journal} {The
  British Journal for the Philosophy of Science},\ }\textbf {\bibinfo {volume}
  {45}},\ \bibinfo {pages} {95} (\bibinfo {year} {1994})}\BibitemShut {NoStop}%
\bibitem [{\citenamefont {Boole}(1862)}]{boole1862theory}%
  \BibitemOpen
  \bibfield  {author} {\bibinfo {author} {\bibfnamefont {G.}~\bibnamefont
  {Boole}},\ }\href@noop {} {\bibfield  {journal} {\bibinfo  {journal}
  {Philosophical Transactions of the Royal Society of London},\ }\textbf
  {\bibinfo {volume} {152}},\ \bibinfo {pages} {225} (\bibinfo {year}
  {1862})}\BibitemShut {NoStop}%
\bibitem [{\citenamefont
  {Papadimitriou}(2003)}]{papadimitriou2003computational}%
  \BibitemOpen
  \bibfield  {author} {\bibinfo {author} {\bibfnamefont {C.}~\bibnamefont
  {Papadimitriou}},\ }\href@noop {} {\emph {\bibinfo {title} {Computational
  complexity}}}\ (\bibinfo  {publisher} {John Wiley and Sons Ltd.},\ \bibinfo
  {year} {2003})\BibitemShut {NoStop}%
\end{thebibliography}%

\end{document}